\documentclass[journal]{IEEEtran}
\ifCLASSINFOpdf
\else
\fi

\usepackage{flushend}
\usepackage{amsmath}
\usepackage{amsthm}
\usepackage{graphicx}
\usepackage{subfigure}
\usepackage{algorithmic}
\usepackage{algorithm}
\usepackage{bbm}
\usepackage{url}

\begin{document}
%
\title{Mobile Instant Video Clip Sharing: \\Modeling and Enhancing View Experience}
%
%
%
\author{Lei~Zhang,~\IEEEmembership{Student~Member,~IEEE,}
        Feng~Wang,~\IEEEmembership{Member,~IEEE,}
        and~Jiangchuan~Liu,~\IEEEmembership{Senior~Member,~IEEE}

\thanks{L. Zhang and J. Liu are with the School of Computing Science, Simon Fraser University, Burnaby, BC, V5A 1S6, Canada (e-mail: lza70@cs.sfu.ca, jcliu@cs.sfu.ca).}
\thanks{F. Wang is with the Department of Computer and Information Science, The University of Mississippi, University, MS, USA, 38677 (e-mail: fwang@cs.olemiss.edu).}
        }
\markboth{Journal of \LaTeX\ Class Files,~Vol.~6, No.~1, January~2007}%
{Shell \MakeLowercase{\textit{et al.}}: Bare Demo of IEEEtran.cls for Journals}
%



\maketitle

\maketitle
\begin{abstract}
With the rapid development of wireless networking and mobile devices,
anytime and anywhere data access becomes readily available nowadays.
Given the crowdsourced content capturing and sharing, the preferred
content length becomes shorter and shorter, even for such multimedia data
as video. A representative is Twitter's Vine service, which, mainly targeting mobile users, enables them to create ultra-short video
clips and instantly post and share with their followers.
In this paper, we present an initial study on this new generation of
instant video clip sharing service enabled by mobile platforms and
explore the potentials towards its further enhancement. We closely investigate its unique mobile interface, revealing
the key differences between Vine-enabled anytime anywhere data access patterns and that of traditional counterparts.
We then examine the scheduling policy to maximize the user watching experience as well as the efficiency on the monetary and energy costs. We
show that the generic scheduling problem involves two subproblems, namely, pre-fetching scheduling and watch-time download scheduling, and
develop effective solutions towards both of them. The superiority of our solution is demonstrated by extensive trace-driven simulations. To the best of our knowledge, this is the first work on modeling and optimizing the instant video clip sharing on mobile devices.
\end{abstract}

%

\begin{IEEEkeywords}
Mobile Device, Video Clip, Instant Video Sharing, Scheduling.
\end{IEEEkeywords}

%
\IEEEpeerreviewmaketitle


\section{Introduction}

In the past decade, the user-generated content (UGC) available on the Internet (e.g., images, videos, micro-blogs, etc.)
has presented an explosive growth trend, which provides significant opportunities for both understanding how users utilize the Internet and enhancing their experience. Video is no doubt a dominant type of media in terms of content generation and sharing. Specifically, for online video sharing, we have witnessed a three-stage evolution with distinct characteristics. The first generation, represented by YouTube~\cite{benevenuto2008understanding, cha2007tube, figueiredo2011tube}, allows users to upload and watch videos directly on a number of video sharing sites (VSSes) over the Internet, in which videos are usually recommended by VSSes or found by search results via search engines. Later, such online social network services as Facebook and Twitter \cite{doman2014event, yan2014mining, yu2014twitter} emerged to directly connect people through cascaded relations, and information thus spreads much faster and more extensively. They greatly changed the earlier video access patterns, too, through proactively and efficiently sharing among friends the video links from external VSSes \cite{rodrigues2011word}. More recently, with the pervasive penetration of wireless mobile networks and the advanced development of smartphones and tablets, video clips can now be ubiquitously generated and accessed~\cite{song2014acceptability}. This leads to the emergence of a new generation video sharing services that focus primarily on capturing and sharing short video clips with mobile users; a representatives is Twitter's Vine\footnote{\url{https://vine.co/}} that was officially launched in the beginning of 2013.

Along with the evolution of video sharing services, users' video watching habits vary significantly: the users of the first generation video sharing services literally browse videos in the Internet through VSS portals or search engines, and in the second generation, the social ties between friends accelerate video propagation among users who have common interests. The ubiquitous mobile access, on which the third generation video sharing services base, allows users to watch videos much more easily and frequently, which in turn enhances the social ties. The migration towards mobile platforms also makes the content length of individual videos being shorter and shorter, from tens of minutes of YouTube video watching to several seconds of instant video clip watching (e.g., 6-second length limit of Vine videos). Instead of having fine-grained VCR controls on the video playback, mobile users typically browse and watch interesting contents by scrolling their mobile device's screen over tons of instant video clips, which are organized in playlists and played seamlessly as the user's viewpoint moves.

Maintaining high-quality user experience during video playback is crucial to all video sharing services, and it becomes an even more challenging task under the mobile scenarios. To put this into perspective, a YouTube video user can tolerate a start-up delay from several seconds to even tens of seconds (advertisements before the playback can be viewed as a kind of delay); an instant video clip itself however is much shorter (e.g., 6-second long for Vine), and a Vine video view can therefore hardly tolerate such typical start-up delays for buffering. Moreover, considering the huge amount of instant video clips consumed during the screen scrolling, the frequent interruptions of single video's playback may largely hurt the user's viewing experience for the whole playlist.

In this paper, we for the first time identify and characterize the unique watching behaviors of the new generation mobile video sharing services, namely, {\it batch view}, {\it passive view}, and {\it screen scrolling}. Using Vine as a representative, we closely investigate its unique mobile interface, revealing the key differences between Vine-enabled anytime anywhere data access patterns and that of traditional counterparts. We then examine the scheduling policy to maximize the user watching experience as well as the efficiency on the monetary and energy costs.  We
show that the generic scheduling problem involves two subproblems, namely, pre-fetching scheduling and watch-time download scheduling, and develop effective solutions towards both of them. We suggest novel popularity-based and wireless-aware pre-fetching for the instant video clips, as well as predicting the watching durations of the approaching videos by such user input gestures as {\em dragging} and {\em flinging}, and then dynamically adapting the schedule for real-time video downloading. Using extensive simulations driven by the real-world traces, we show that our solution can significantly improve the user watching experience while still keeping both the monetary and energy costs relatively low.

The rest of the paper is organized as follows: Section~\ref{sec:back} presents an overview of the new generation of mobile instant video clip sharing services, which reshape the way that mobile users watch videos, and discusses the motivation of our modeling work. In Section~\ref{sec:ps}, we formulate the generic scheduling problem. By dividing it into two subproblems and conquering them separately, we then propose a general solution for the generic problem. Section~\ref{sec:ev} evaluates our solution by extensive trace-driven simulations and a further discussion on other open issues is given in Section~\ref{sec:fd}. Finally, we conclude our paper in Section~\ref{sec:con}.

\section{Background and Motivation}
\label{sec:back}

We now take a close look at Vine, a representative of the new generation mobile instant video clip sharing services. The social tie in Vine is the typical {\em follower-followee} relationship, as in Twitter and other similar social network services, e.g., Instagram and Flickr. The users of Vine can follow others whom they are interested in, and then receive updates from them. They can also access the public channels to view promoted videos or videos on specific topics. A typical mobile Vine client has four pages: Home/Feed, Explore, Activity and Profile. As illustrated in Figure \ref{fig:intf}, a user can view, like, comment, and share (repost) recent posts from others that it follows in  the ``Home/Feed'' page; in the ``Explore'' page, popular/trending posts, as well as dedicated public channels, are highlighted, together with the interface for the user to search specific videos or people of interest. The other two pages are not designed for video watching: the ``Activity'' page acts as the notification center showing the recent events, and the ``Profile'' page allows a user to customize his/her personal settings.


\begin{figure}[tbp]
\centering
\subfigure[\small Home page]{
        \includegraphics[scale=0.17]{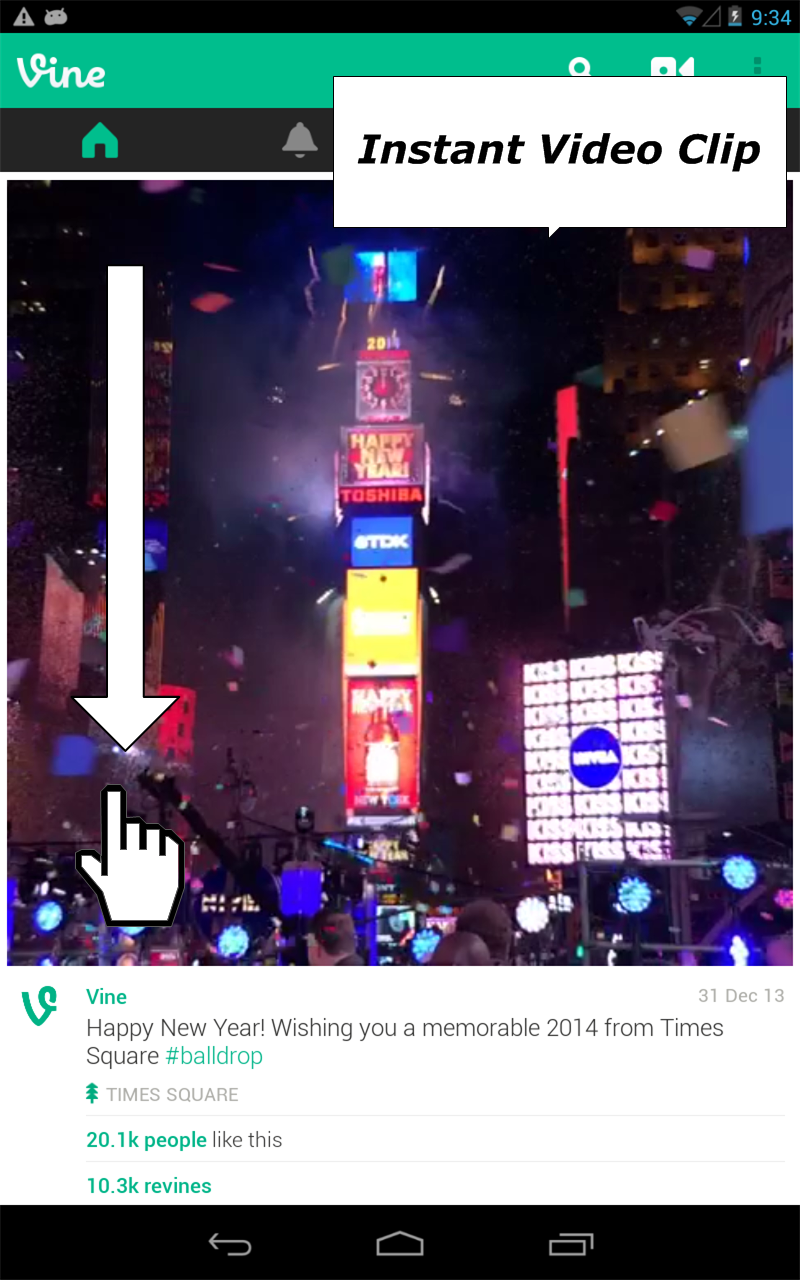}
\label{fig:intf1}
}
\subfigure[\small Explore page]{
        \includegraphics[scale=0.17]{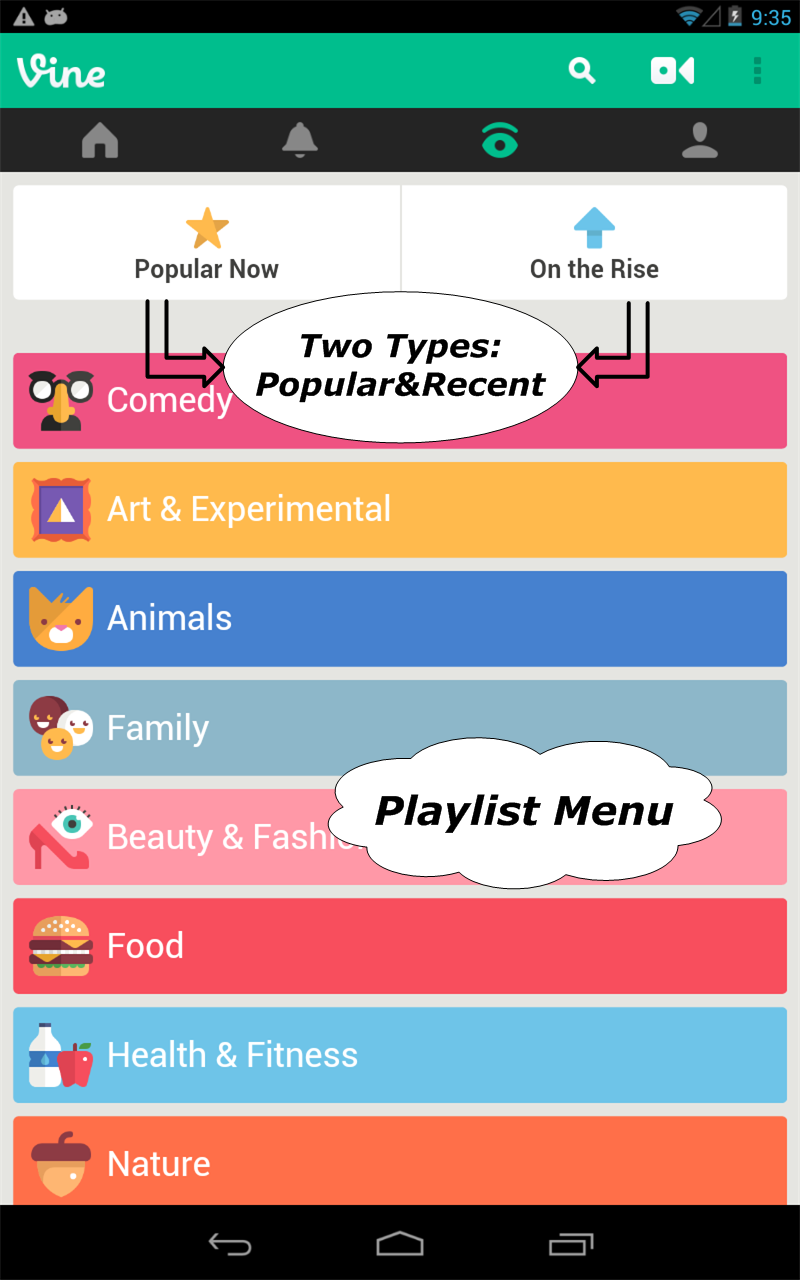}
\label{fig:intf2}
}
\caption{Vine's interfaces} \vspace{-0.4cm}
\label{fig:intf}
\end{figure}

Vine preserves similar social structure as in other online social networks. A key (and significant) difference is that the media of interest in Vine are now ultra short video clips (with 6 second limit for now), while not the traditional texts or still images. This makes its user experience notably different from that of the traditional mobile social networking and video sharing services. In traditional VSSes and OSNs, the users need to click on the interested item to view or link to a specific video, which only allows users to view one video each time/click. Vine-like services, however, return a playlist of video clips once the user decides to view the updates for certain users, tags, or channels. As the user scrolls the screen of his smartphone/tablet, the mobile app plays a number of video clips from the generated list seamlessly. Given the linear formation of instant video clips in playlists, \textit{screen scrolling} is the key user action for the mobile instant video clip sharing services, which uniquely differentiates Vine-like services not only from traditional VSSes and OSNs, but also from other mobile VoD or video streaming applications. The small file size of the instant video clip makes it very suitable for the mobile scenario in terms of downloading and watching: the user can easily access the service at anytime and anywhere, and watch a number of instant video clips even for a short duration (e.g., waiting for a bus or during a course break).

We refer to this unique user behavior of viewing multiple video clips with the screen scrolling as {\em Batch View}. The batch view behavior implies that mobile users can watch a considerable amount of instant video clips within the playback time of one normal video such as a YouTube video. Besides batch view, another distinct feature is {\em Passive View}, as the mobile users have to ``passively'' watch some of the video clips in the playlist. Similar to a Twitter-like user interface, which is widely adopted in other social networking services, the selected instant video clips that will be viewed are arranged in order, and users have no control over the order of the playbacks. If the user is only interested in two specific video clips that are separated in the playlist, s/he may have to watch all the video clips between them. The only way to proactively skip an instant video clip is to scroll the screen so fast that there is not enough time to play it, which however can hardly be well controlled in practice. If the user scrolls the screen slowly, the following ultra short video clip may have been downloaded and started to play, and thus the user will probably watch it. As the 6-second playback time is so short that it is even not enough to finish reading the description of some video clips (users may have already watched the video clip before making the decision based on the description), the Vine users tend to passively view numerous video clips. The two key user behaviors abstract and emphasize the unique features of this new generation of mobile video clip sharing service.

\begin{figure}[tbp]
        \centering
                \includegraphics[scale=0.425]{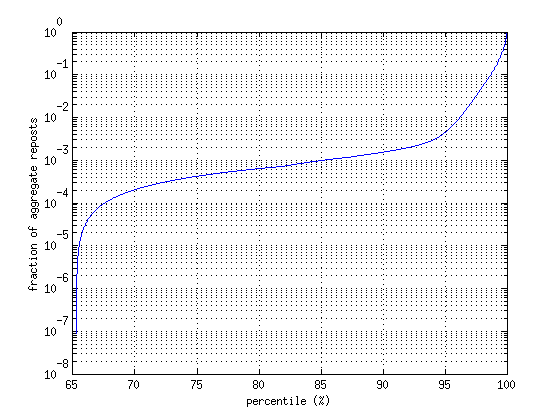}
        \caption{Popularity skewness of Vine instant video clips} \vspace{-0.4cm}
        \label{fig:clog}
\end{figure}

The mobile instant video clip sharing services often lack VCR controls (such as rewind and fast forwarding) that are common in traditional video playback; rather, the mobile users can scroll the screen over the playlist to watch different instant video clips. This mode of instant video clip watching seems to be inflexible, but has its unique advantages in the mobile scenarios. First, as the instant video clips are so short, there is no need to have close and complex controls over each of them. Moreover, given the limited display size of each video, screen scrolling is effective in approaching successive instant video clips in the playlist, enabling users to find interesting contents more easily, and accelerates the propagation of popular videos. Yet, as the key of this video watching mode, screen scrolling, if not being handled properly, may ruin the viewing experience. In the worst case (e.g. a user downloads every instant video clip through 3G with a very limited bandwidth), a vicious circle can happen: the downloading of a just skipped video will take up the network resources and block the downloading of the one that the user intends to watch, which will in turn force the user to give up watching the target video and scroll forward to search for other interesting videos. Conventional approaches, such as deploying proxy servers, may not deal with this issue well, as the bottleneck still exists at the last hop where mobile devices have to download a considerable amount of instant video clips. On the other hand, as a well-known fact, in traditional video sharing services, user interests always focus on a small portion of popular contents. We examine the popularity distribution of Vine videos as shown in Figure~\ref{fig:clog} (details about the data collection will be introduced in Section \ref{sec:ev}), where the top 5\% popular Vine video clips account for more than 99\% reposts. Compared to the first two generations of video sharing services, Vine shows a much higher skewness: the top 10\% popular YouTube videos account for nearly 80\% of views \cite{cha2007tube}; whereas the top 2\% videos in Renren take up 90\% of the total requests, and the top 5\% videos attract 95\% of requests \cite{li2012video}. This extreme skewness implies the opportunities for pre-fetching popular contents at mobile terminals before the video watching. We accordingly model the unique instant video clip watching mode, and enhance the user experience from the mobile side. To the best of our knowledge, this is the first work on modeling and optimizing the instant video clip sharing from this perspective.

\section{Download and Watch Scheduling: Problem and Solution}
\label{sec:ps}

\begin{figure}[tbp]
        \centering
                \includegraphics[scale=0.45]{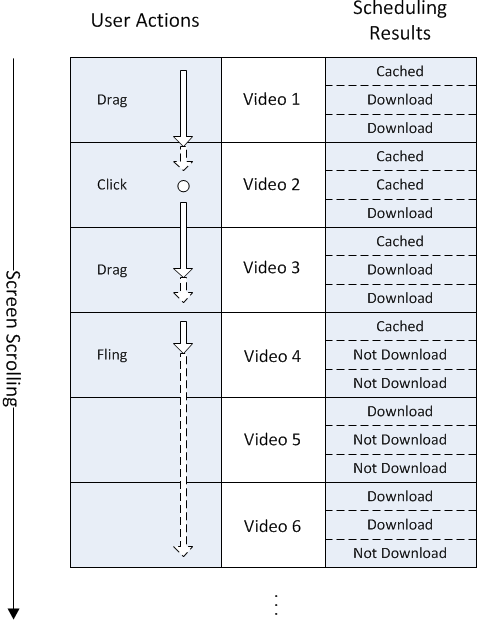}

        \caption{An illustration of playlist optimization} \vspace{-0.4cm}
        \label{fig:list}
\end{figure}

As mentioned, Vine videos are organized in different playlists, which can be characterized into three types: the list of video updates from followees (social videos), the list of promoted videos in popular sections (popular videos), and the list of user uploaded videos in recent sections (recent videos). Only the playlist of social videos changes with different users, and the other two types of playlists remain the same across users.
Denote the playlist of instant video clips that will be watched by the user as $V=\{v_1,v_2,...,v_n\}$. As illustrated in Figure \ref{fig:list}, according to different user input actions, each instant video clip may remain in the user's viewport for different duration. We use $U=\{u_1,u_2,...,u_n\}$ to denote such durations, where $u_i$ corresponds to the duration that the user watches the video $v_i$. Also, we let $u_0$ denote the time that the user starts watching the playlist. We consider two types of network connections in this formulation: mobile connection (e.g, 3G) and wireless connection (e.g., WiFi). We use $B(t)$, $C(t)$, and $E(t)$ to denote the available bandwidth, the monetary cost, and the energy consumption at a given time $t$, respectively, where $B(t) \in \{B_{wifi}, B_{3G}\}$, $C(t) \in \{C_{wifi}, C_{3G}\}$ ($C_{wifi}=0$, since the cost for WiFi connections is usually negligible), and $E(t) \in \{E_{wifi}, E_{3G}\}$. As in practice both the major mobile communication protocols and wireless communication protocols work in some basic time units, in our problem, we also divide the time evenly into discrete time slots, where one time slot may contain an integer multiple of such basic time units.
Let $R$ be the video streaming rate of the mobile video sharing service, and $L$ be the video length (in Vine's case, 6 seconds). It is worth noting that given the video length limit is only several seconds, most users create and post videos that fulfill the length limit and thus the file size of Vine videos is almost the same after transcoding to a certain resolution. Therefore, we use the same $R$ and $L$ for different instant video clips. Define a video downloading schedule as $S =\{(\hat{v_1},\hat{t_1},\hat{l_1}),(\hat{v_2},\hat{t_2},\hat{l_2}),...,(\hat{v_k},\hat{t_k},\hat{l_k})\}$, where a tuple $(\hat{v_i},\hat{t_i},\hat{l_i})$ ($\hat{v_i} \in V$ and $ \hat{l_i}>0$) means at time $\hat{t_i}$, we start to download video $\hat{v_i}$ for the duration $\hat{l_i}$.

Our problem is to find a proper schedule $S$, that can optimize the video watching experience with high efficiency. To this end, we define the playback discontinuity of a single video $v_i$ watched for the duration $u_i$ as:


\begin{equation*}
\begin{aligned}
& discontinuity(v_i)=1-\frac{1}{\min(u_i,L)} \cdot \sum_{\substack{t\in(\sum\limits_{k=0}^{i-1}u_k,\min(\sum\limits_{k=0}^{i}u_k,\\ \sum\limits_{k=0}^{i-1}u_k+L)]}} \\
&\mathbbm{I}{\displaystyle [\sum_{\substack{(\hat{v}_j,\hat{t}_j,\hat{l}_j)\in S,\\ \hat{v}_j=v_i, \hat{t}_j\leq t}}\sum_{\hat{t} = \hat{t}_j}^{\min(\hat{t}_j+\hat{l}_j-1,t)} B(\hat{t}) \geq (t-\sum_{k=0}^{i-1}u_k)\cdot R]},
\end{aligned}
\end{equation*}
where $\mathbbm{I}{[\cdot]}$ is the indicator function.
The single video playback discontinuity naturally reflects the user experience for a continuous playback, which calculates how many time slots the downloading of this video misses the deadline for the playback. We further define the playback discontinuity of the playlist $V$ as a weighted sum of the single video playback discontinuities:
\begin{equation*}
Discontinuity= \sum_{v_i\in V} w_i \cdot discontinuity(v_i),
\end{equation*}
where $w_i$ is the normalized weight for $v_i$. An intuitive assignment of $w_i$ can be $\frac{1}{\sum_{k=1}^n u_k} u_i$,
which assigns higher weights to the videos that have longer watching durations, as longer watching durations usually imply higher user interests. We will further discuss more specific assignments of $w_i$ later.

Our objective is thus to minimize the playback discontinuity, as well as maximize the efficiency, which is to minimize the total monetary cost:
\begin{equation*}
C_{total}=\sum_{(\hat{v}_j,\hat{t}_j,\hat{l}_j)\in S} \sum_{t=\hat{t}_j}^{\hat{t}_j+\hat{l}_j-1} C(t),
\end{equation*}
and the total energy consumption:
\begin{equation*}
E_{total}=\sum_{(\hat{v}_j,\hat{t}_j,\hat{l}_j)\in S} \sum_{t=\hat{t}_j}^{\hat{t}_j+\hat{l}_j-1} E(t).
\end{equation*}
It is easy to see that these objectives may contradict with each other, as downloading more portions of the playlist can reduce the playback discontinuity but will also inevitably consume more energy and may increase the monetary expense. We thus adopt the following linear combination form to align them together:
\begin{equation}
\label{equ:obj}
p \cdot Discontinuity + q \cdot \frac{C_{total}}{C_{max}} + r \cdot \frac{E_{total}}{E_{max}},
\end{equation}
where $p$, $q$ and $r$ are the parameters to assign different weights to the three goals. As $Discontinuity$ is a ratio between $[0,1]$, we also normalize the monetary cost and the energy consumption by their corresponding maximum values, where $C_{max}$ is the maximum total cost of the case that all the videos in the playlist are downloaded through 3G links, and $E_{max}$ can be obtained similarly. We then have the following theorem:
\newtheorem{theorem}{Theorem}
\begin{theorem}
The decision version of the modeled generic downloading scheduling problem is a NP-complete problem.
\end{theorem}

\begin{proof}
We first state the corresponding decision problem: given all the required parameters, is there a schedule with the objective value given by Equation \ref{equ:obj} at most $M$? Given an instance of this decision problem, a certificate that it is solvable would be a specification of the downloading schedules for each video. We could then check each video's playback discontinuity, downloading monetary cost, energy consumption and whether the objective value is no greater than $M$.

We now show that the Subset Sum problem is reducible to this decision problem. Consider an instance of Subset Sum with numbers $\hat{l}_1,...,\hat{l}_n$ and a target $W$. To construct an equivalent scheduling instance, one may be struck initially by the fact that we have so many parameters to manage. The key is to sacrifice some of the flexibility, producing a simpler ``skeletal" instance of the problem that still encodes the Subset Sum problem. Assume that $p/q$ and $p/r$ are large enough so that any increase of discontinuity will increase the objective value. Let $B_{wifi}\geq B_{3G} \geq R$, and thus, for each video, the playback discontinuity can be guaranteed to be zero as long as it is scheduled for downloading no later than its playback. We also assume that for any time slot, $C_{3G} > C_{wifi}$, $E_{3G} > E_{wifi}$. Given the type of connection, the monetary cost and energy consumption are fixed for each time slot. Therefore, the objective value given by Equation \ref{equ:obj} can be taken as a function of number of time slots for a certain type of connection (defined as $O_{wifi}(m)$ and $O_{3G}(m)$, where $m$ is the number of time slots).

Given videos $1,2,...,n$ and video $k$ has a watch duration of $u_k$, Define $\hat{l}_k$ as the time slots taken to download video $k$ ($\hat{l}_k=u_k\cdot R/B$, where $B \in \{B_{wifi}, B_{3G}\}$).
The corresponding full Subset Sum instance is to find whether there is a non-empty subset that sum to $W$ given a set of integers $\hat{l}_1, \hat{l}_2,...,\hat{l}_n$.
Assume that there are $W$ time slots of WiFi links before the first video playback, and during the playback from video 1 to video $n$, only 3G connection is available. Given any value of $M$, $W$ can be determined such that
$O_{wifi}(W)+O_{3G}((R \cdot \sum_{k=1}^n u_k - W\cdot B_{wifi})/B_{3G})=M$. 

Now consider any feasible schedule to this instance of the decision problem. To limit the objective value no greater than $M$, there must not be any idle time slot of WiFi downloadings before the first video playback according to the definition of $W$.
If the first $W$ WiFi slots have not been fully utilized, there must be some videos that have not been downloaded or more than appropriate amount of videos are downloaded through 3G connection instead of using WiFi, both of which would cause the objective value greater than $M$ according to our assumptions.
In particular, if videos $i_1,...,i_k$ are the ones that are scheduled to be downloaded in the first $W$ WiFi slots with durations $\hat{l}_{i_1},...,\hat{l}_{i_k}$, then the corresponding numbers $\hat{l}_{i_1},...,\hat{l}_{i_k}$ in the Subset Sum instance add up to exactly $W$.

Conversely, if there are numbers $\hat{l}_{i_1},...,\hat{l}_{i_k}$ that add up to exactly $W$, then we can schedule these intervals to be download videos in the first $W$ WiFi time slots, and the remainder are downloaded during the playback, which is a feasible solution to the decision problem. This finishes the proof that the decision version of the original optimization problem is NP-complete.
\end{proof}

We now present a practically effective solution that takes the unique structure of the problem into account. It divides the problem into two cases depending on  whether the scheduling part is before or during a user watches the playlist. The former demands a pre-fetching scheduling, whereas the latter demands a watch-time download scheduling.

\subsection{Pre-fetching Scheduling}
 We first consider the pre-fetching scheduling. It happens well before the user starts watching the playlist, i.e., without a stringent time constraint; hence we can offload the mobile traffic to the wireless network to reduce the transmission cost. The objective is to find a schedule $S_{pf}$ to pre-fetch the videos, subjecting to the following constraints:\\

\noindent (1) Storage Constraint:
\begin{equation*}
\sum_{(\hat{v}_j,\hat{t}_j,\hat{l}_j)\in S_{pf}} \sum_{t=\hat{t}_j}^{\hat{t}_j+\hat{l}_j-1} B(t)\leq StorageSize ;
\end{equation*}
\\
(2) Cost Constraint:
\begin{equation*}
\forall (\hat{v}_j,\hat{t}_j,\hat{l}_j)\in S_{pf},
\sum_{t=\hat{t}_j}^{\hat{t}_j+\hat{l}_j-1} C(t) = 0 .
\end{equation*}
The storage constraint ensures that the total amount of pre-fetched video will not exceed the limited local storage. And the cost constraint implies that the pre-fetching is performed only through WiFi links. Note that the playlist $V$ during the pre-fetching may only be a subset of that during the watch-time, as the pre-fetching occurs before the video watching and new videos may be added to the playlist after the pre-fetching, which will be handled by our watch-time download scheduling to be discussed in the next subsection.

Given that the user behavior during the video watching is unknown at this stage (nor the watching duration $U$), we thus introduce $P=\{p_1,p_2,...,p_n\}$ to denote the user preference on each video in the playlist, which can reflect the potential lengths of the watch durations. In practice, $P$ can be determined in different ways: for example, it may be a specific metric such as video popularity, video timeliness or the social distance between the publisher (the user who reposts the video) and the consumer (the user who may watch the video); it can be also modeled as a function which considers all of the above factors. In this work, we use the video popularity as the metric of user preference. In addition, we introduce parameter $\alpha \in [0,1]$ to represent the aggressiveness of the pre-fetching. In particular, for each instant video clip, we pre-fetch $\alpha$ of the total video, instead of downloading the whole clip. Given the limited storage on mobile devices, pre-fetching a part of the video, not all of it, indicates that more videos can be pre-fetched, and thus the start-up delay can be greatly reduced, which also makes the pre-fetching scheduling can better work with the watch-time download scheduling.
The playback discontinuity of a single video $v_i$ can then be rewritten as
\begin{equation*}
discontinuity(v_i) = 1-\frac{1}{\alpha \cdot L} \cdot  \frac{pf(v_i)}{R},
\end{equation*}
where $pf(v_i)$ defines how much of this video has been pre-fetched:
\begin{equation*}
pf(v_i)=\sum_{(\hat{v}_j,\hat{t}_j,\hat{l}_j)\in S_{pf}, \hat{v}_j=v_i}\sum_{\hat{t}=\hat{t}_j}^{\hat{t}_j+\hat{l}_j-1}B(\hat{t}).
\end{equation*}

As stated in the generic formulation, the next step is to find a proper assignment of $w_i$. For this subproblem, we define $w_i$ as
\begin{equation*}
w_i=\frac{1}{\sum_{k=1}^n p_k \cdot discontinuity(v_k)} p_i \cdot discontinuity(v_i),
\end{equation*}
which considers both the user preference for $v_i$ and its current playback discontinuity. It is worth noting that $w_i$ decreases as more of $v_i$ has been pre-fetched. The intuition is that, given the batch view behavior, it is not reasonable to allocate all the resources to a tiny portion of extremely popular videos. In practice, the first several units of a video clip are usually requested with a much higher probability than its later part. Together with the pre-fetching aggressiveness $\alpha$, this assignment of $w_i$ allows us to pre-fetch more instant video clips with the video preference still being considered.

As the monetary cost for WiFi links is usually negligible, our goal here is to minimize $Discontinuity$ with the following form:
\begin{equation*}
\begin{aligned}
Discontinuity 
= & \frac{1}{\sum_{k=1}^n p_k \cdot discontinuity(v_k)}\\
& \cdot \sum_{v_i\in V} p_i \cdot discontinuity(v_i)^2.
\end{aligned}
\end{equation*}
One may notice that, different from Equation \ref{equ:obj}, this objective function does not directly involve the energy consumption of pre-fetching. In this formulation, we use $\alpha$ to control the trade-off between the energy consumption and the playlist playback discontinuity. As $\alpha$ gets larger, more videos would be pre-fetched, which causes more energy consumption; on the contrary, if $\alpha$ is small, only a small portion of videos will be pre-fetched, and thus little energy is consumed. Therefore, the above objective function can still represent the overall performance.

This pre-fetching scheduling subproblem is a variation of the knapsack problem with a total weight limit
\begin{equation*}
W=\min(StorageSize, \sum_{\forall t \text{ such that } C(t)=0}B(t)),
\end{equation*}
where an object is one time slot length of video playback, and its value is the amount of decrease of $p_i \cdot discontinuity(v_i)^2$ after pre-fetching one more time slot, if the object belongs to video $v_i$. It is easy to see that while the weight of each object is the same, the value changes as the decisions are made, i.e., as one object of video $v_i$ is downloaded, the value of all the remained objects of video $v_i$ decreases as now $discontinuity(v_i)$ decreases. To this end, we design a greedy algorithm by searching and downloading one object that currently has the greatest value in each iteration. Recall that in our scenario all the objects have the same weight. Therefore, as our solution can find the optimal result in each iteration, it can actually return the final optimal pre-fetching schedule.


\subsection{Watch-time Download Scheduling}
Unlike the pre-fetching, during the watch-time, the video watching durations can be largely determined from the input user gestures. In general, the mobile instant video clip sharing services allow three types of user gestures: click, drag and fling, where the last two gestures can cause screen scrolling. Once an user gesture is given, the following process of screen moving is predetermined. Given the fixed display size of each instant video clip (specifically, the fixed height), the motion of screen scrolling can be modeled and calculated, and the details of the scrolling process can be obtained accurately (e.g., how many videos are present, how long each video will stay in the viewport), which can hardly be done in VoD or video streaming applications. Although different operating systems have different technical details for implementation, the philosophy for animating the screen scrolling is generally the same, which is to gradually decelerate the scrolling speed until it reaches zero if there is no other finger touch detected during the deceleration. We next show how to calculate the video watching durations according to user gestures, by taking the Android OS as an example.

By detecting and collecting relevant information about the user's finger touch, the initial scrolling speed $s_0$ can be calculated as the dragging distance divided by the touch time in the unit of $pixels/second$. As the screen only scrolls vertically in mobile video sharing services, we denote the height of each instant video clip as $h$. By default, the threshold $s_T$ for the initial scrolling speed to distinguish between a drag and a fling in the Android OS is $50$ $pixels/second$, which can be scaled under different configurations based on the actual screen resolution.

\textit{Dragging:}
In the case of dragging, the screen scrolling speed will experience a uniform deceleration with the default deceleration $d=2000$ $pixels/second^2$. Given the initial speed $s_0$, there will be $\lfloor s_0^2/2hd \rfloor$ video clips covered by this drag gesture. In the deceleration process, for the $m$-th video clip showing in the dragging animation, we have
\begin{equation}
\label{equ:d1}
mh=s_0t_m-dt_m^2/2,
\end{equation}
where $t_m$ is the time that the $(m+1)$-th video clip starts to enter the viewport ($t_0=0$). Solving Equation \ref{equ:d1} gives us
\begin{equation}
\label{equ:d2}
t_m=(s_0-\sqrt{s_0^2-2mhd})/d  \text{   }  (1 \leq m \leq \lfloor s_0^2/2hd \rfloor).
\end{equation}

\textit{Flinging:}
If a fling is detected, the deceleration will change with the scrolling speed. Given the scrolling speed $s$, the total fling duration $T(s)$ and the total fling distance $D(s)$ can be calculated by using the following equations:
\begin{equation}
\label{equ:f1}
l(s)=log[0.35\cdot s/(Fric\cdot P_{COEF})],
\end{equation}
\begin{equation}
\label{equ:f2}
T(s)=1000\cdot exp[l(s)/(D_{RATE}-1)],
\end{equation}
\begin{equation}
\label{equ:f3}
D(s)=Fric\cdot P_{COEF}\cdot exp[D_{RATE}/(D_{RATE}-1)\cdot l(s)],
\end{equation}
where $D_{RATE}=log(0.78)/log(0.9)$, $Fric$ denotes the parameter of the friction with the default value as $0.015$, and $P_{COEF}=G\cdot 39.37\cdot ppi\cdot 0.84$. To compute $P_{COEF}$, $G$ is the gravity of Earth with a constant value of $9.80665$ $m/s^2$, $39.37$ is used for the conversion between meters and inches, and $ppi$ denotes the parameter of pixels per inch for the specific mobile device. From Equation \ref{equ:f2} and \ref{equ:f3}, we can derive
\begin{equation}
\label{equ:f4}
D(s)=Fric\cdot P_{COEF}\cdot (T(s)/1000)^{D_{RATE}},
\end{equation}
and
\begin{equation}
\label{equ:fadd1}
T(s)=1000\cdot [D(s)/(Fric\cdot P_{COEF})]^\frac{1}{D_{RATE}}.
\end{equation}
Given the initial speed $s_0$, there will be $\lfloor D(s_0)/h \rfloor$ video clips covered by this fling gesture. Assume $s_m$ is the scrolling speed at time $t_m$. In the deceleration process, the following equation is also satisfied:
\begin{equation}
\label{equ:f5}
D(s_0)-D(s_m)=mh \text{		}  (1 \leq m \leq \lfloor D(s_0)/h \rfloor).
\end{equation}
By combining Equation \ref{equ:f4}, \ref{equ:fadd1}, and \ref{equ:f5}, we have

\begin{equation}
\begin{aligned}
\label{equ:f6}
t_m&=T(s_0)-T(s_m) \\
&=T(s_0)-1000\cdot [D(s_m)/(Fric\cdot P_{COEF})]^\frac{1}{D_{RATE}} \\
&=T(s_0)-1000\cdot [(D(s_0)-mh)/(Fric\cdot P_{COEF})]^\frac{1}{D_{RATE}} \\
&=T(s_0)-1000\cdot [D(s_0)/(Fric\cdot P_{COEF})-mh/(Fric\cdot P_{COEF})]^\frac{1}{D_{RATE}} \\
&=T(s_0)-1000\cdot [(T(s_0)/1000)^{D_{RATE}}-mh/(Fric\cdot P_{COEF})]^\frac{1}{D_{RATE}} \\
&(1 \leq m \leq \lfloor D(s_0)/h \rfloor)
\end{aligned}
\end{equation}



It is worth noting that, as the basis of this analysis, Equation \ref{equ:d1}, \ref{equ:f1}, \ref{equ:f2} and \ref{equ:f3} are obtained from our analysis of the Android OS source code.
Assume that the user will focus on one video at any given time. From Equations~\ref{equ:d2} and~\ref{equ:f6}, the watching duration of $m$-th video clip showing in the screen scrolling animation can be obtained as $u_m=t_m-t_{m-1}$.
Based on the above analysis, we can now tell how many videos are scrolled by a user gesture and how long each video can stay in the viewport. Therefore, the video watching duration $U=\{u_1,u_2,...,u_n\}$ is available once the input user gestures are given. This subproblem of watch-time download scheduling is thus to find a proper real-time download schedule $S_{rd}$, so as to minimize our objective:
\begin{equation*}
p \cdot Discontinuity + q \cdot \frac{C_{total}}{C_{max}} + r \cdot \frac{E_{total}}{E_{max}}.
\end{equation*}

\begin{algorithm}[tbh]
\caption{Playlist Scheduling}
\label{alg1}
\begin{algorithmic}[1]
\WHILE {\TRUE}
\IF {a video $v$ is shown on screen}

\STATE {Remove $v$ from $V_p$;}
\IF {a new user input gesture comes}
\STATE {Set all slots in $Q$ to empty;}
\STATE {Obtain $V_w$ based on the user input gesture model;}

\STATE Sort $V_w$ in descendant order according to $u_i^2 \cdot discontinuity(v_i)$;
\STATE Delete the videos that cannot reduce the objective value if downloaded from $V_w$;
\WHILE {$V_w$ is not empty}
\STATE Pick the 1st video $v_1$ out of $V_w$;
\WHILE {$v_1$ is not fully downloaded}
\IF {there exist available slots in $Q$}
\STATE Update $Q$ to assign the closest-to-deadline slot to $v_1$;
\ELSE
\STATE break;
\ENDIF
\ENDWHILE
\ENDWHILE
\IF {there are available WiFi slots in $Q$}
\STATE Update $Q$ to move later scheduled downloadings forward to fill the WiFi slots;
\ENDIF
\ENDIF

\ELSE
\STATE Update $V_p$ for newly arrived videos;

\IF {cache is not full \AND WiFi is available}
\STATE Search $V_p$ to find the video $v$ with the largest decrease amount of $p_i \cdot discontinuity(v_i)^2$ assuming one more unit of the video is pre-fetched;
\STATE Update $Q$ to schedule downloading one unit of $v$;

\ENDIF
\ENDIF

\STATE Download one unit of video if currently scheduled in $Q$ and update $Q$ accordingly;

\ENDWHILE
\end{algorithmic}
\end{algorithm}

Note that $Discontinuity$ here needs to consider the result of the pre-fetching schedule with the updated $discontinuity(v_i)$:
\begin{equation*}
\begin{aligned}
& discontinuity(v_i)=1-\frac{1}{\min(u_i,L)} \cdot \sum_{\substack{t\in(\sum\limits_{k=0}^{i-1}u_k,\min(\sum\limits_{k=0}^{i}u_k,\\ \sum\limits_{k=0}^{i-1}u_k+L)]}} \\
&\mathbbm{I}\displaystyle [\sum_{\substack{(\hat{v}_j,\hat{t}_j,\hat{l}_j)\in S_{rd},\\ \hat{v}_j=v_i, \hat{t}_j\leq t}}\sum_{\hat{t} = \hat{t}_j}^{\substack{\min(\hat{t}_j+\\\hat{l}_j-1,t)}} B(\hat{t}) + pf(v_i) \geq (t-\sum_{k=0}^{i-1}u_k)\cdot R],
\end{aligned}
\end{equation*}
where the amount of the video $v_i$ that has been pre-fetched ($pf(v_i)$) is also considered in the calculation. For the watch-time downloading subproblem, we define $w_i$ as
\begin{equation*}
w_i=\frac{1}{\sum_{k=1}^n u_k^2} u_i^2,
\end{equation*}
which emphasizes the importance of watch durations.

The problem is essentially to find a best trade-off between playback discontinuity, monetary cost and energy consumption during the watch-time. Given the objective function, we are able to tell whether downloading one unit of $v_i$ is beneficial to the overall result. Therefore, only those videos that can reduce the objective value if downloaded need to be considered for scheduling. As mentioned in the generic problem, videos with higher watching durations have higher impacts on the objective function, which should be scheduled for downloading with higher priorities.
To this end, we propose a heuristic that always tries to download one time slot of the video $v_i$ that has the highest value of $u_i^2 \cdot discontinuity(v_i)$ in the unscheduled set, and schedules its downloading interval as late as possible (i.e., closest to but still before its playback deadline) so as to only introduce the minimal impact on other videos to be scheduled later. If there still exist available WiFi downloading slots after the scheduling finishes, we move the video downloadings scheduled in the later time slots forward to these WiFi slots so as to fill all the time slots with lower cost.

We integrate and summarize our solutions for both subproblems in Algorithm~\ref{alg1}. $V_p$ is the set of videos that are considered for pre-fetching, which can be initialized as the whole playlist. $V_w$ is the set of videos that are scheduled during watch-time according to the dragging/flinging model given an user gesture input. $Q$ is the schedule queue, which denotes whether each time slot is available or assigned to download one unit of a certain video. We analyze the time complexity of Algorithm~\ref{alg1} by examining the two parts that solve the the corresponding subproblems, respectively. For the watch-time downloading subproblem, the first part of our algorithm (line 2-22) first sorts $V_w$ accordingly ($O(|V_w|\log|V_w|)$), and then search $Q$ for a proper schedule for each video ($O(|V_w||Q|)$). Although this part is executed every time a user gesture is detected, given that any user gesture can only affect a limited number of videos that will show on the screen, $|V_w|$ is hence quite small, and so is the searching space in $Q$, which implies that our algorithm can work efficiently during the watch-time.
For the pre-fetching subproblem, the second part of our algorithm (line 23-28) searches for the next video to pre-fetch. Assume that each video is divided into $M$ units. The searching process will thus be performed at most $|M||V_p|$ rounds. As each round of the search process only has the complexity of $O(|V_p|)$, its efficiency is also acceptable especially given that the pre-fetching happens well before the video watching and is often with a much longer time span.

\section{Performance Evaluation}
\label{sec:ev}

\begin{figure}[btp]
	\centering
		\includegraphics[width=0.325\textwidth]{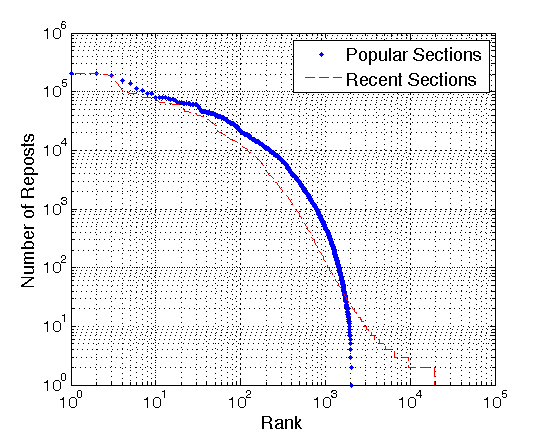}
	\caption{Popularity distribution}
\vspace{-0.4cm}
	\label{fig:pop}
\end{figure}

\begin{figure}[btp]
\centering
\subfigure[\small Videos uploaded into popular sections]{
	\includegraphics[width=0.22\textwidth]{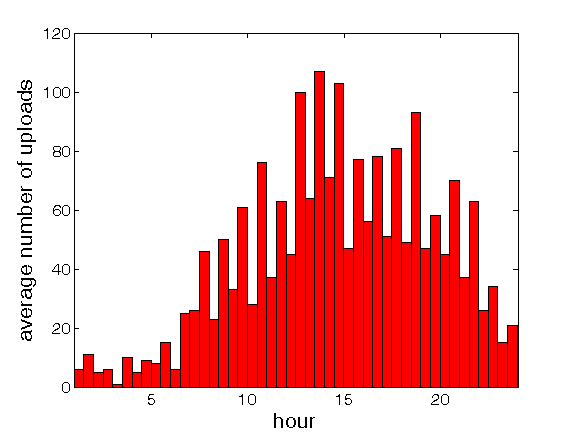}
\label{fig:vp1}
}
\subfigure[\small Videos uploaded into recent sections]{
	\includegraphics[width=0.22\textwidth]{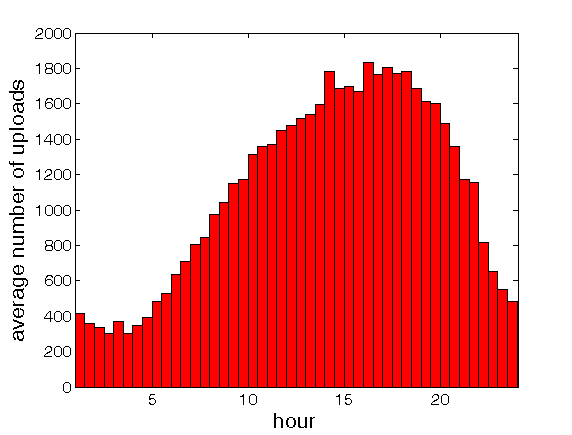}
\label{fig:vp2}
}
\caption{Average video uploading rate}
\vspace{-0.4cm}
\label{fig:vp}
\end{figure}

In this section we evaluate our proposed approach by conducting extensive trace-driven simulations, which exploit the real-world data of Vine videos and the user gestures recorded during the Vine watching experience.

\subsection{Data Sets and Traces}
\subsubsection{Vine Video Measurement}
Our measurement on Vine is mainly on the video uploading rate and the video popularity in terms of the number of reposts. In particular, we develop a customized crawler that imitates the behaviors of the mobile application by sending the same requests and substituting the header information with the headers of mobile browsers. The collected data set consists of the videos from popular sections and recent sections in Vine's user channels during the week of November 21-27, 2013, where we recorded the exact time that each instant video clip was uploaded to the user channels. We further track how many times each video has been reposts as the metric of popularity. Figure \ref{fig:pop} plots the video popularity versus its rank for the two types of videos, in which videos from recent sections exhibit a higher skewness. The average numbers of hourly video uploads during a day for popular and recent sections are shown in Figure \ref{fig:vp}.

\subsubsection{User Gesture Traces}
As we have emphasized in previous sections, Vine-like services introduce the revolutionary change of video watching mode for mobile users. To obtain the real-world user gesture traces that directly reflect the unique user experience of the new generation of mobile instant video clip sharing services, we implement a mobile app on Andriod platforms to record the exact user touch events. We recruit 10 volunteers to watch Vine videos using the official client on Android under 3G connections and WiFi connections, respectively. Each experiment is conducted around 5 minutes. We plot two important characteristics of user behavior from the collected traces: Figure \ref{fig:int} shows the probability density distribution of the user gestures' inter-arrival time,  and its curve fitting result; and Figure \ref{fig:fs} plots the histogram of the measured initial scrolling speed of the triggered flinging animations. The user gesture traces are used to simulate video watch durations by applying the dragging/flinging model.

\begin{figure}[tbp]
\centering
\subfigure[\small Inter-arrival time of user actions under 3G scenarios]{
	\includegraphics[width=0.2225\textwidth]{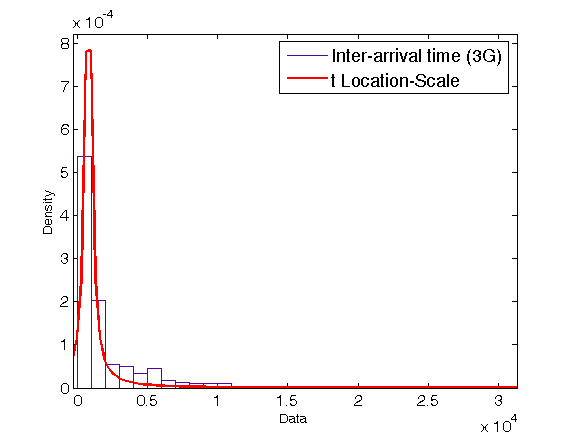}
\label{fig:int1}
}
\subfigure[\small Inter-arrival time of user actions under WiFi scenarios]{
	\includegraphics[width=0.2225\textwidth]{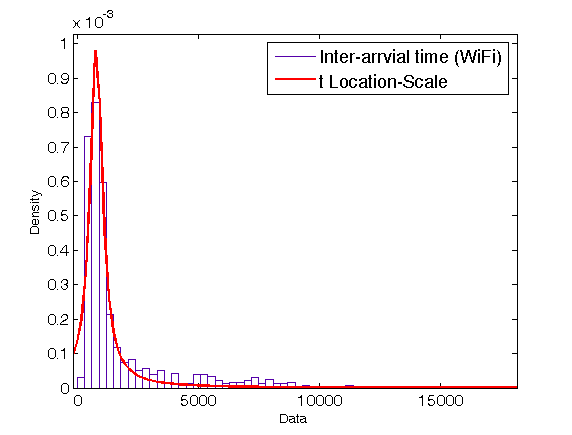}
\label{fig:int2}
}
\caption{Density of inter-arrival time of user gestures}
\vspace{-0.4cm}
\label{fig:int}
\end{figure}

\begin{figure}[tbp]
\centering
\subfigure[\small Fling speed under 3G scenarios]{
	\includegraphics[width=0.2225\textwidth]{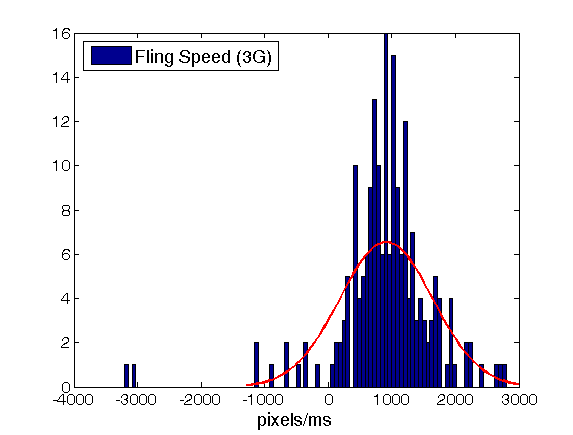}
\label{fig:fs1}
}
\subfigure[\small Fling speed under WiFi scenarios]{
	\includegraphics[width=0.2225\textwidth]{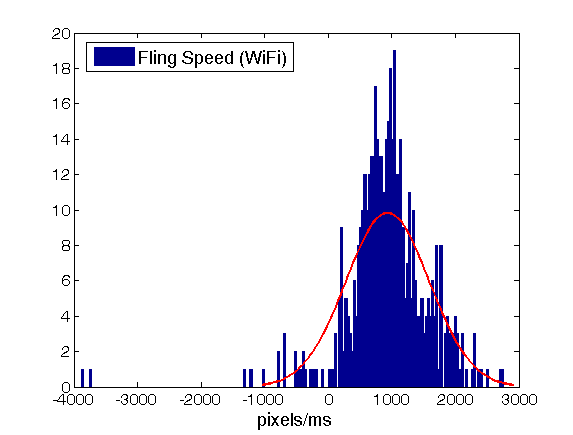}
\label{fig:fs2}
}
\caption{The initial scrolling speed of flings}
\vspace{-0.4cm}
\label{fig:fs}
\end{figure}

\subsection{Methodology}
We conduct extensive simulations to evaluate our proposed pre-fetching and watch-time downloading solution with the first-hand measurement results as well as data traces from Vine and the mobile users who participated in our data collection experiments. For comparison, we implement another two video downloading schemes.
\textit{Sequential Downloading} (SeqD) downloads all the videos according to their order in the playlist, and disregards all the user actions, which takes the playlist with instant video clips as a single long video, and is the most aggressive downloading scheme with the least flexibility.
\textit{Next-one Downloading} (NextD) always attempts to download the next video that will enter the viewport during the current playback, which emulates the current caching strategy of Vine. If the current watching duration is not long enough to download an instant video clip, the next playback will be interrupted. Unlike our proposed solution, both of SeqD and NextD are not capable to obtain and utilize the information of watching durations, which implies that they always try to download the entire target video rather than a fraction of it.
In addition, we use playback discontinuity, monetary cost, energy consumption as the three metrics in our evaluation, which directly relate to the three goals of our optimization problem. We randomly introduce 20 video watching events of Vine during the daytime from 9 a.m. to 9 p.m. in the week-long data set, emulate the mobile user's video watching behaviors by using the user gesture traces, and produce the average results. Each video watching event consumes 50 videos from popular sections and 150 videos from recent sections. We assume that, during the watch-time only 3G links are available with 1 MB/s bandwidth, and WiFi is available for pre-fetching once every hour. The application local storage is set as 100 MB, which is about half of the total amount of video consumption in one watching event, given the file size of an instant video clip is around 1 MB.
The monetary cost model is 1 dollar per 100 MB traffic. The energy model is adopted from \cite{balasubramanian2009energy}.

\subsection{Impacts of $p/q$ and $p/r$ Ratios}

\begin{figure}[tbp]
\centering
\subfigure[\small Impact of different $p/q$ values on efficiency and playback discontinuity ($r=0$)]{
	\includegraphics[width=0.2225\textwidth]{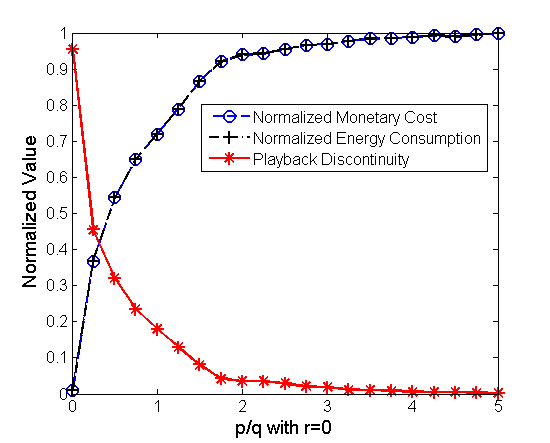}
\label{fig:para1}
}
\subfigure[\small Impact of different $p/r$ values on efficiency and playback discontinuity ($q=0$)]{
	\includegraphics[width=0.2225\textwidth]{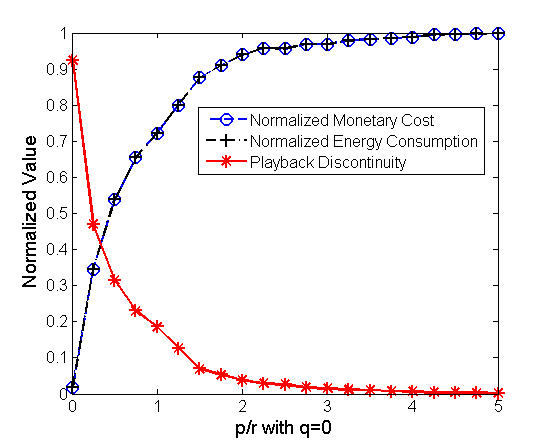}
\label{fig:para2}
}
\caption{Impacts of $p$, $q$ and $r$} \vspace{-0.4cm}
\label{fig:para}
\end{figure}

We first investigate the impact of different parameter settings by varying the values of $p$, $q$ and $r$. As we have three goals in our objective function, two of which are designed for efficiency, we vary the ratio of $p/q$ with $r=0$, and vary the ratio $p/r$ with $q=0$, respectively. The results are shown in Figures~\ref{fig:para1} and \ref{fig:para2}, which demonstrate how the playback discontinuity, the cost efficiency and the energy efficiency change with different parameter values. As $p$, $q$ and $r$ only affect our optimization goal during the watch-time, we disable pre-fetching in this experiment, and only focus on studying the proper settings of $p/q$ and $p/r$ for watch-time downloading (WT). Since changing the ratios between $p$, $q$ and $r$ is essentially the trade-off between efficiency and performance, it is not surprising that the plots in Figures~\ref{fig:para1} and \ref{fig:para2} are very similar. Note that, as we adopt both a linear monetary cost model and a linear energy model, the normalized results for monetary cost and energy consumption is nearly the same when pre-fetching has been disabled. When $p/q$ ($p/r$) is small ($<0.5$), very limited number of videos have been downloaded and the video playback is severely affected, so as to reduce the monetary cost (energy consumption). On the other hand, if $p/q$ ($p/r$) increases to a certain degree ($>3$), the video watching experience is optimized but results in a much higher monetary cost (energy consumption). Moreover, there is a small interval near 1.5 in the figure, in which the playback quality is acceptable with a relatively good cost efficiency (energy efficiency). We thus pick $p/q=1.5$ and $p/r=1.5$ as the default setting for the remaining evaluation. It is worth noting that where the best trade-off achieves for the ratio of $p/q$ ($p/r$) may change when different monetary cost model and energy consumption model are applied.
Figure~\ref{fig:para} implies that we can have close controls over the performance of our watch-time scheduling scheme by varying the values of $p/q$ and $p/r$, which can be flexible and adjusted based on the user needs in real-time. For example, if a user watches Vine videos at home or in office, s/he probably does not care about the cost and wants to maximize the watching experience by setting $p/q$ and $p/r$ to be large values; while, for a user who accesses Vine service with a limited data plan or a low battery, small values should be assigned to $p/q$ and $p/r$, so that the monetary cost/energy consumption for the video watching can be reduced.
%

\subsection{Impact of Pre-fetching Aggressiveness $\alpha$}

\begin{figure}[tbp]
\centering
\subfigure[\small Relative monetary cost efficiency and playback discontinuity of WT+PF normalized by WT]{
	\includegraphics[width=0.2225\textwidth]{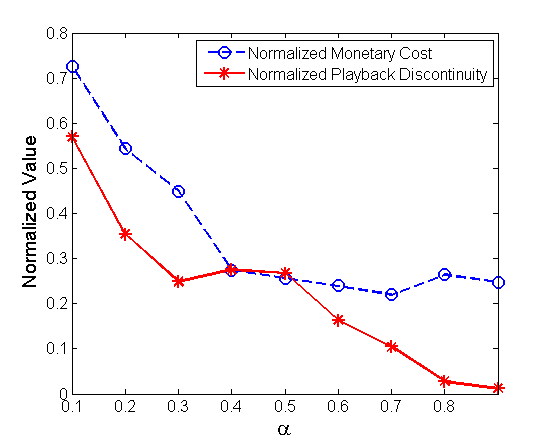}
\label{fig:alpha1}
}
\subfigure[\small Relative energy efficiency of WT+PF normalized by WT]{
	\includegraphics[width=0.2225\textwidth]{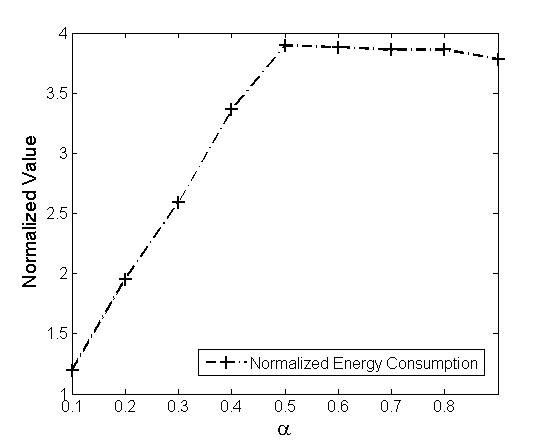}
\label{fig:alpha2}
}
\caption{Impact of $\alpha$} \vspace{-0.4cm}
\label{fig:alpha}
\end{figure}

We next study how the overall performance of our proposed solution changes with $\alpha$, which, as discussed in the previous section, represents how aggressively our pre-fetching (PF) acts. We vary the value of $\alpha$ from 0.1 to 0.9, and compare the major performance metrics of our proposed solution with pre-fetching (WT+PF) and those of watch-time downloading without pre-fetching (WT). We normalize the performance metrics of WT+PF by using WT as the baseline, and plot them in Figure~\ref{fig:alpha}. As shown in Figure~\ref{fig:alpha1}, when $\alpha$ grows large, the playback discontinuity decreases gradually, while the monetary cost first decreases and then becomes stable.
The reason is that when the pre-fetching becomes more and more aggressive, more videos that will be watched are downloaded through WiFi, thus the monetary cost is saved; however, as one cannot accurately predict which video (and which portion of it) will be consumed and the application local storage is also limited, there are still some parts of videos need to be downloaded during the watch-time.
At the same time, as more and more videos are pre-fetched, some of the videos, which originally will not be downloaded during watch-time according to the performance-efficiency trade-off, are pre-fetched and make the playback discontinuity continue to decrease.
On the other hand, Figure~\ref{fig:alpha2} shows that the normalized energy consumption keeps growing (from 1.2x to 3.8x) until $\alpha$ reaches 0.5, where the application local storage is used up and limits the amount of pre-fetched videos even if $\alpha$ gets larger than 0.5 (and thus the energy consumption). It is worth noting that here the comparison is only between WT+PF and WT, not WT+PF and SeqD (WT+PF and NextD), where, as will be discussed in the next subsection, the energy saving can achieve as much as over 90\%.
Figure~\ref{fig:alpha} suggests us setting $\alpha=0.2$ as the default value, where WT+PF consumes about 2x energy of WT to save around 45\% monetary cost and improve the playback discontinuity by over 60\%.

\subsection{Performance Enhancement}

Mobile users can access the instant video clip sharing services under any circumstances. We next show the overall performance gain of our proposed solution (WT+PF) compared to Sequential Downloading (SeqD) and Next-one Downloading (NextD) under different operating conditions. To simulate different levels of connection quality, we vary the downloading bandwidth from 150 KB/s to 3 MB/s.
Figure~\ref{fig:b1} shows the results with $p/q=1.5$ and $p/r=1.5$, which plots the playback discontinuity in Figure~\ref{fig:b1pd}, and the normalized monetary cost and the normalized energy consumption in Figure~\ref{fig:b1nc}.
As we keep $\alpha$ fixed in this experiment, the monetary cost and the energy consumption exhibit similar patters due to their linear models, and thus we see overlapping plots for NextD in Figure~\ref{fig:b1nc}.
As SeqD naively downloads videos and disregards their playbacks, which no doubt introduces the highest monetary cost/energy consumption, we use it as the baseline, and normalize the costs of the other two approaches.
As shown in the figure, our proposed approach (WT+PF) has a very stable performance with different downloading bandwidths, in terms of both playback discontinuity and cost efficiency. On the contrary, the playback discontinuity of NextD and that of SeqD increasing dramatically when the bandwidth is low ($<$ 0.5 MB/s). The reason is that, as the bandwidth becomes lower, it is more and more difficult for these two downloading schemes to finish each downloading before the playback, while our proposed approach can still keep the playback discontinuity at a low level by smartly managing the downloading according to user actions and efficient pre-fetching.
On the other hand, the monetary cost/energy consumption of NextD quickly increases after reaching the minimal at 0.5 MB/s.
The reason for both NextD and SeqD introducing high monetary cost/energy consumption is that, when the bandwidth is low, both NextD and SeqD keep the downloading link busy almost all the time; whereas if the bandwidth becomes high enough, both of the downloading schemes attempt to download all the videos, especially given the fact that, both NextD and SeqD disregard the video watch duration and download the entire video instead of part of it. Figure~\ref{fig:b1nc} shows that our proposed approach can save at least over 40\% monetary cost and 30\% energy consumption, and the cost/energy saving under high bandwidths can be higher than 90\%.
This result implies that our solution can efficiently offload the traffic from the watch-time downloading to the pre-fetching, and reduce a significant amount of unnecessary watch-time downloadings, as it only downloads the part of the video that will be watched given the input user actions. One may notice that the other two caching schemes may achieve the similar performance to our proposed approach in terms of playback discontinuity, when the bandwidth is high enough so that most downloadings can be simply finished before the playback deadline. This is because the setting of $p/q$ and $p/r$ asks for a balance between cost efficiency and playback discontinuity.
As previously discussed, if a user is more aggressive on the video watching experience, s/he can further increase of the ratio of $p/q$ and $p/r$, e.g., to $p/q=3.5$ and $p/r=3.5$. We plot the corresponding results in Figure~\ref{fig:b2}, which shows that our solution can always achieve the best playback discontinuity at diverse bandwidths, and still with huge amounts of (over 90\%) cost/energy savings. To sum up, the results demonstrate that, compared to the other two caching schemes, our proposed approach achieves much more stable performance, which is a crucial requirement for practical designs.

\begin{figure}[tbp]
\centering
\subfigure[\small Playback discontinuity]{
	\includegraphics[width=0.2225\textwidth]{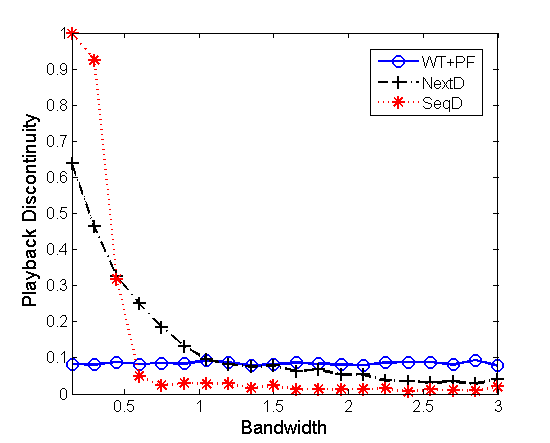}
\label{fig:b1pd}
}
\subfigure[\small Monetary and energy costs normalized by the baseline results of SeqD]{
	\includegraphics[width=0.2225\textwidth]{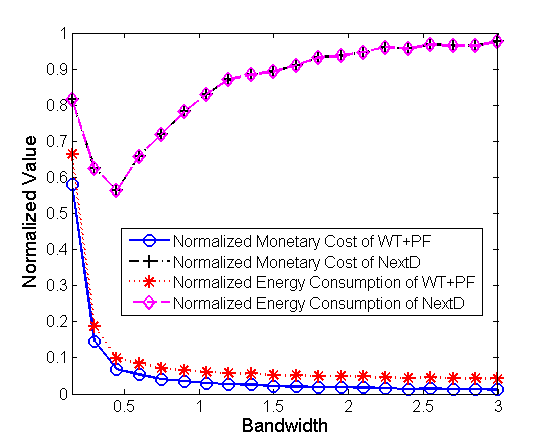}
\label{fig:b1nc}
}
\caption{Impact of downloading bandwidth when $p/q=1.5$, $p/r=1.5$} \vspace{-0.4cm}
\label{fig:b1}
\end{figure}

\begin{figure}[tbp]
\centering

\subfigure[\small Playback discontinuity]{
	\includegraphics[width=0.2225\textwidth]{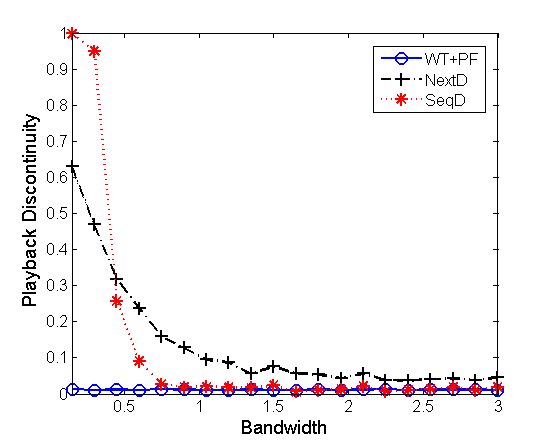}
\label{fig:b2pd}
}
\subfigure[\small Monetary and energy costs normalized by the baseline results of SeqD]{
	\includegraphics[width=0.2225\textwidth]{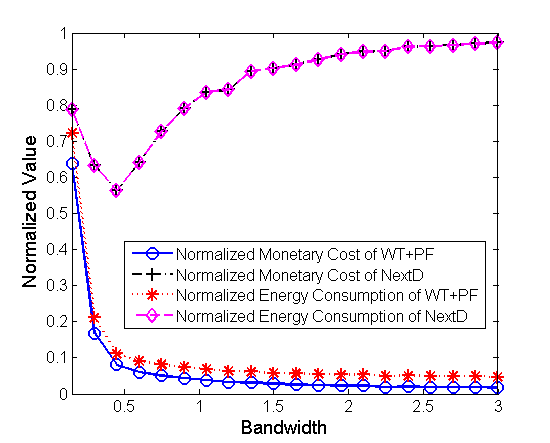}
\label{fig:b2nc}
}
\caption{Impact of downloading bandwidth when $p/q=3.5$, $p/r=3.5$} \vspace{-0.4cm}
\label{fig:b2}
\end{figure}

\section{Further Discussion} \label{sec:fd}

In this section, we briefly discuss several open issues that are worth further exploring. The increasing availability of various sensors integrated in mobile devices provides new opportunities for understanding the operating conditions of users and the surrounding environment. First, extracting patterns from mobile users' daily activities can assist us to find proper locations for pre-fetching. Specifically, for a mobile user, there are a number of locations that s/he would stay as a daily routine, such as office, classroom and home. These routine locations, which normally provide dedicated WiFi connections, can act as regular pre-fetching sites, as the WiFi access is normally free and users may also be able to charge their mobile devices at these places. Besides the free daily routine WiFi connections, mobile users are also exposed to unforeseen wireless networks from time to time (e.g., when having coffee at a cafe with free WiFi), where opportunistic pre-fetching can be performed.

Moreover, understanding and predicting the operating conditions allows us to make full use of the opportunities that would be normally ignored by mobile users. Existing works have been done on energy-efficient mobile context sensing \cite{lu2010jigsaw,nath2012ace}. For example, Acquisitional Context Engine \cite{nath2012ace} infers user's current context by dynamically learning relationships among various context attributes. One simple rule that has been learned by ACE can be: if the user is static, WiFi is connected, and the phone is muted, then the user is probably having a meeting, in which case opportunistic pre-fetching can be done without affecting other applications. Further integrating such sensing abilities with a human mobility model based on periodic travels \cite{cho2011friendship} can bring us the chances to make our proposed approach even smarter, i.e, automatically adjusting parameter values of $p$, $q$ and $r$ according to the currently obtained operating context.

Finally, our work has also touched the interests of some other research fields such as popularity prediction and human-computer interaction. A number of typical solutions of popularity prediction for online contents have been widely studied in the literature, most of which focus on predicting the trend based on time series with regression models \cite{szabo2010predicting,wang2012guiding} and classification models \cite{shamma2011viral,yang2011patterns}. On the other hand, measurement studies \cite{cheng2008statistics,li2012video} also show that the related statistics of video sharing services may be well fit by certain distributions. In our work, video popularity is used to estimate the level of user interests in the future accesses. Although numerous papers have studied how users express interests by examining and understanding various user behaviors \cite{white2009predicting,banerjee2009user}, we may push it one step further -- predicting future user behaviors based on the potential user interests, where machine learning techniques can be applied.

\section{Conclusion} \label{sec:con}
In this paper, we presented the first initial study on the new
generation of instant video clip sharing services enabled by mobile
platforms and explored the potentials for its further enhancement. Taking Vine as an
example, we closely investigated its mobile user interface,
identified and characterized the unique watching behaviors of the
new generation mobile video sharing services, namely, batch view,
passive view and screen scrolling. We formulated a generic scheduling
problem to maximize the viewing experience as well as the
efficiency on the monetary and energy costs. We showed that the formulated generic problem is NP-complete and to better solve it, we further divided the problem
into two subproblems, specifically, the pre-fetching scheduling
problem and the watch-time download scheduling problem, conquered them
separately and then developed a general solution for the generic
problem. Using extensive simulations driven by the real-world traces, we showed that our solution can significantly improve
the viewing experience while still keeping both the monetary
and energy costs relatively low.

\bibliographystyle{IEEEtran}
\bibliography{IEEEabrv,references}

\begin{thebibliography}{10}
\providecommand{\url}[1]{#1}
\csname url@samestyle\endcsname
\providecommand{\newblock}{\relax}
\providecommand{\bibinfo}[2]{#2}
\providecommand{\BIBentrySTDinterwordspacing}{\spaceskip=0pt\relax}
\providecommand{\BIBentryALTinterwordstretchfactor}{4}
\providecommand{\BIBentryALTinterwordspacing}{\spaceskip=\fontdimen2\font plus
\BIBentryALTinterwordstretchfactor\fontdimen3\font minus
  \fontdimen4\font\relax}
\providecommand{\BIBforeignlanguage}[2]{{%
\expandafter\ifx\csname l@#1\endcsname\relax
\typeout{** WARNING: IEEEtran.bst: No hyphenation pattern has been}%
\typeout{** loaded for the language `#1'. Using the pattern for}%
\typeout{** the default language instead.}%
\else
\language=\csname l@#1\endcsname
\fi
#2}}
\providecommand{\BIBdecl}{\relax}
\BIBdecl

\bibitem{benevenuto2008understanding}
F.~Benevenuto, F.~Duarte, T.~Rodrigues, V.~A. Almeida, J.~M. Almeida, and K.~W.
  Ross, ``Understanding video interactions in {Y}ou{T}ube,'' in \emph{ACM MM},
  2008.

\bibitem{cha2007tube}
M.~Cha, H.~Kwak, P.~Rodriguez, Y.-Y. Ahn, and S.~Moon, ``I tube, you tube,
  everybody tubes: analyzing the world's largest user generated content video
  system,'' in \emph{ACM IMC}, 2007.

\bibitem{figueiredo2011tube}
F.~Figueiredo, F.~Benevenuto, and J.~M. Almeida, ``The tube over time:
  characterizing popularity growth of {Y}ou{T}ube videos,'' in \emph{ACM WSDM},
  2011.

\bibitem{doman2014event}
K.~Doman, T.~Tomita, I.~Ide, D.~Deguchi, and H.~Murase, ``Event detection based
  on {T}witter enthusiasm degree for generating a sports highlight video,'' in
  \emph{ACM MM}, 2014.

\bibitem{yan2014mining}
M.~Yan, J.~Sang, and C.~Xu, ``Mining cross-network association for {Y}ou{T}ube
  video promotion,'' in \emph{ACM MM}, 2014.

\bibitem{yu2014twitter}
H.~Yu, L.~Xie, and S.~Sanner, ``Twitter-driven {Y}ou{T}ube views: Beyond
  individual influencers,'' in \emph{ACM MM}, 2014.

\bibitem{rodrigues2011word}
T.~Rodrigues, F.~Benevenuto, M.~Cha, K.~Gummadi, and V.~Almeida, ``On
  word-of-mouth based discovery of the web,'' in \emph{ACM IMC}, 2011.

\bibitem{song2014acceptability}
W.~Song, D.~Tjondronegoro, and I.~Himawan, ``Acceptability-based {Q}o{E}
  management for user-centric mobile video delivery: A field study
  evaluation,'' in \emph{ACM MM}, 2014.

\bibitem{li2012video}
H.~Li, H.~Wang, J.~Liu, and K.~Xu, ``Video sharing in online social networks:
  measurement and analysis,'' in \emph{ACM NOSSDAV}, 2012.

\bibitem{balasubramanian2009energy}
N.~Balasubramanian, A.~Balasubramanian, and A.~Venkataramani, ``Energy
  consumption in mobile phones: a measurement study and implications for
  network applications,'' in \emph{ACM IMC}, 2009.

\bibitem{lu2010jigsaw}
H.~Lu, J.~Yang, Z.~Liu, N.~D. Lane, T.~Choudhury, and A.~T. Campbell, ``The
  jigsaw continuous sensing engine for mobile phone applications,'' in
  \emph{ACM Sensys}, 2010.

\bibitem{nath2012ace}
S.~Nath, ``{ACE}: exploiting correlation for energy-efficient and continuous
  context sensing,'' in \emph{ACM MobiSys}, 2012.

\bibitem{cho2011friendship}
E.~Cho, S.~A. Myers, and J.~Leskovec, ``Friendship and mobility: user movement
  in location-based social networks,'' in \emph{ACM KDD}, 2011.

\bibitem{szabo2010predicting}
G.~Szabo and B.~A. Huberman, ``Predicting the popularity of online content,''
  \emph{Communications of the ACM}, vol.~53, no.~8, pp. 80--88, 2010.

\bibitem{wang2012guiding}
Z.~Wang, L.~Sun, C.~Wu, and S.~Yang, ``Guiding internet-scale video service
  deployment using microblog-based prediction,'' in \emph{IEEE INFOCOM}, 2012.

\bibitem{shamma2011viral}
D.~A. Shamma, J.~Yew, L.~Kennedy, and E.~F. Churchill, ``Viral actions:
  Predicting video view counts using synchronous sharing behaviors.'' in
  \emph{AAAI ICWSM}, 2011.

\bibitem{yang2011patterns}
J.~Yang and J.~Leskovec, ``Patterns of temporal variation in online media,'' in
  \emph{ACM WSDM}, 2011.

\bibitem{cheng2008statistics}
X.~Cheng, C.~Dale, and J.~Liu, ``Statistics and social network of {Y}ou{T}ube
  videos,'' in \emph{IEEE/ACM IWQoS}, 2008.

\bibitem{white2009predicting}
R.~W. White, P.~Bailey, and L.~Chen, ``Predicting user interests from
  contextual information,'' in \emph{ACM SIGIR}, 2009.

\bibitem{banerjee2009user}
N.~Banerjee, D.~Chakraborty, K.~Dasgupta, S.~Mittal, A.~Joshi, S.~Nagar,
  A.~Rai, and S.~Madan, ``User interests in social media sites: an exploration
  with micro-blogs,'' in \emph{ACM CIKM}, 2009.

\end{thebibliography}

\end{document}